\documentclass[12pt]{article}

\usepackage{amsmath}
\usepackage{amssymb}
\usepackage{amsfonts}
\usepackage{latexsym}
\usepackage{bm}
\usepackage{color}

\catcode `\@=11 \@addtoreset{equation}{section}

\catcode `\@=12

\newcommand{\be}{\begin{equation}}
\newcommand{\en}{\end{equation}}
\newcommand{\bea}{\begin{eqnarray}}
\newcommand{\ena}{\end{eqnarray}}
\newcommand{\beano}{\begin{eqnarray*}}
\newcommand{\enano}{\end{eqnarray*}}
\newcommand{\bee}{\begin{enumerate}}
\newcommand{\ene}{\end{enumerate}}

\newcommand{\mb}{\mathbb}

\newcommand{\mc}{\mathcal}

\newcommand{\D}{{\mc D}}

\newcommand{\E}{{\cal E}}
\newcommand{\F}{{\cal F}}

\newcommand{\Hil}{\mc H}

\newtheorem{thm}{Theorem}[section]

\newtheorem{prop}[thm]{Proposition}

\newenvironment{proof}{\noindent {\bf Proof --}}{\hfill$\square$ \vspace{3mm}\endtrivlist}

\newcommand{\ip}[2]{\left\langle {#1}\left|\right.{#2} \right\rangle}
\newcommand{\balpha}{{\mbox{\boldmath${\alpha}$}}}
\newcommand{\bgamma}{{\mbox{\boldmath${\gamma}$}}}

\newcommand{\bbeta}{{\mbox{\boldmath${\beta}$}}}
\newcommand{\balphabar}{{\overline{\mbox{\boldmath${\alpha}$}}}}
\newcommand{\bgammabar}{{\overline{\mbox{\boldmath${\gamma}$}}}}
\newcommand{\bbetabar}{{\overline{\mbox{\boldmath${\beta}$}}}}
\catcode `\@=11 \@addtoreset{equation}{section}
\catcode `\@=12

\begin{document}

\thispagestyle{empty}

\vspace*{2cm}

\begin{center}
{{\Large \bf Non-self-adjoint hamiltonians defined by Riesz bases }}\\[10mm]


{\large F. Bagarello} \footnote[1]{ Dipartimento di Energia, Ingegneria dell'Informazione e Modelli Matematici,
Facolt\`a di Ingegneria, Universit\`a di Palermo, I-90128  Palermo, and INFN, Universit\`a di di Torino, ITALY\\
e-mail: fabio.bagarello@unipa.it\,\,\,\, Home page: www.unipa.it/fabio.bagarello}
\vspace{3mm}\\

{\large A. Inoue} \footnote[2]{Department of Applied Mathematics, Fukuoka University, Fukuoka 814-0180, Japan\\
e-mail: a-inoue@fukuoka-u.ac.jp}
\vspace{3mm}\\

{\large C. Trapani}\footnote[3]{Dipartimento di Matematica e Informatica, Universit\`a di Palermo, I-90123 Palermo, Italy\\e-mail: camillo.trapani@unipa.it}
\end{center}

\vspace*{2cm}

\begin{abstract}
\noindent We discuss some features of non-self-adjoint Hamiltonians with real discrete simple spectrum under the assumption that the eigenvectors form a Riesz basis of Hilbert space. Among other things, {we give conditions under which these Hamiltonians} can be factorized in terms of generalized lowering and raising operators.

\end{abstract}


\vfill


\newpage

\section{Introduction}

In the recent literature an increasing interest has been devoted to non-self-adjoint Hamiltonians, mainly in connections with {\em quasi-hermitian quantum mechanics}, and its relatives. A very extensive literature has been produced in the past ten years, mostly physically oriented, \cite{mosta}. However, also some mathematically-minded results have been obtained,  as, for instance, in \cite{siegl,fb1,bagnew, jpact}.  In \cite{mosta2}, in particular, the author has considered non-hermitian hamiltonians with real spectrum. However, in our opinion, more attention should be paid to the fact that an hamiltonian operator, self-adjoint or not, is usually an unbounded operator. This means that domain problems, and not only, arise out of this feature. In this paper, continuing an analysis already began some years ago, \cite{bagpbrep,bit2011,bit2012} we consider this aspect of the theory, taking into account the unbounded nature of several operators appearing in the game.

Let $H$ be a (non necessarily self-adjoint) {closed} operator defined on a dense subset $D(H)$ of the Hilbert space $\Hil$, with inner product
$\ip{\cdot}{\cdot}$, linear in the first entry, and related norm $\|.\|$.

We assume here that $H$ has purely discrete simple spectrum { (this means that the spectrum, $\sigma(H)$, consists only of isolated eigenvalues
 with multiplicity one)}. Even though the interesting case for physical applications is $\sigma(H)\subset {\mb R}$ we
 will not make this restrictive assumption but we will suppose that the corresponding  eigenvectors form a Riesz basis, \cite{nota1},
 $\F_\phi=\{\phi_n,\,n\geq0\}$ for $\Hil$, see \cite{rieszbasis}.
This means that there exist an orthonormal basis (ONB, for short) $\E=\{e_n,\, n\geq 0\}$ and a bounded operator $T$,
invertible and with bounded inverse $T^{-1}$ such that $\phi_n=Te_n$ for all $n\geq0$. Let us now define $\psi_n=(T^{-1})^* e_n$ and $\F_\psi=\{\psi_n,\,n\geq0\}$. It is clear that $\F_\psi$ is a Riesz basis, too. Moreover, it is biorthogonal to $\F_\phi$: $\ip{\phi_n}{\psi_m}=\delta_{n,m}$, and
$$
f=\sum_{n=0}^\infty\ip{f}{\phi_n}\psi_n=\sum_{n=0}^\infty \ip{f}{\psi_n}\phi_n,\quad
\forall f \in \Hil.$$

Furthermore, the operators $S_\phi$ and $S_\psi$ defined by
$$
S_\phi f=\sum_{n=0}^\infty \ip{f}{\phi_n}\phi_n,\qquad S_\psi f=\sum_{n=0}^\infty \ip{f}{\psi_n}\psi_n,
$$
are bounded, everywhere defined  in $\Hil$, positive and, therefore, self-adjoint. Also, they are one the inverse of the other: $S_\phi=(S_\psi)^{-1}$ and $S_\phi=TT^*$.

{Actually, the operators $S_\psi$ intertwines $H$ and $H^*$ (and, of course, $S_\phi$ intertwines $H^*$ and $H$). These operators are, in fact {\em similar} and $H$ is {\em quasi-Hermitian} \cite{jpctsmall}, that is,
 $S_\psi H$ is a symmetric operator, i.e., $S_\psi H \subseteq (S_\psi H)^*$,
 and it is then natural to pose the question if $H$ is quasi-self-adjoint; i.e.the equality $S_\psi H = (S_\psi H)^*$ holds. The latter condition, in turn, is equivalent to saying that $H$ is  similar to a self-adjoint operator.}

In order to answer to this and other questions, we perform in this paper (Section 2) a detailed analysis of operators defined  by Riesz bases. In particular we characterize domains and adjoints; moreover we construct intertwining operators in general form and show that, if the eigenvalues are real, then the operator can be made self-adjoint in a different Hilbert space, {more explicitly, in the same space but endowed with a different inner product}. A similar analysis has been carried out by many authors before, see, for instance, \cite{mosta2}. However, in our knowledge, this is the first systematic and mathematically {minded}  study of the problem.  In Section 3, following our previous work \cite{bit2012} we construct and study generalized raising and lowering operators, and we discuss the possibility of factorizing a given hamiltonian.
Section 4 is devoted to examples and applications, while our conclusions are given in Section 5.

\section{Some operators defined by Riesz bases}
Let $\{\phi_n\}$ be a Riesz basis in $\Hil$ and, as above, $\psi_n= (T^{-1})^*e_n$, $n=0,1,\ldots\,$. As remarked before, $\{\psi_n\}$ is also a Riesz basis and it is biorthogonal to $\{\phi_n\}$; i.e. $\ip{\phi_n}{\psi_m}=\delta_{nm}$.

 Throughout this section let $\balpha=\{\alpha_n\}$ be any sequence of complex numbers. We define two operators \cite{nota2}
$$H_{\phi, \psi}^{\balpha}= \sum_{n=0}^\infty \alpha_n \phi_n \otimes \overline{\psi}_n$$
and
$$H_{\psi, \phi}^\balpha=\sum_{n=0}^\infty \alpha_n \psi_n \otimes \overline{\phi}_n$$
as follows:

$$ \left\{\begin{array}{l} D(H_{\phi, \psi}^\balpha)=\left\{f \in \Hil; \sum_{n=0}^\infty \alpha_n \ip{f}{\psi_n}\phi_n \mbox{ exists in }\Hil \right\}\\
{ } \\
H_{\phi, \psi}^\balpha f= \sum_{n=0}^\infty \alpha_n \ip{f}{\psi_n}\phi_n, \; f \in D(H_{\phi, \psi}^\balpha)
\end{array}\right.
$$

$$ \left\{\begin{array}{l} D(H_{\psi, \phi}^\balpha)=\left\{f \in \Hil; \sum_{n=0}^\infty \alpha_n \ip{f}{\phi_n}\psi_n \mbox{ exists in }\Hil \right\}\\
{ } \\
H_{\psi, \phi}^\balpha f= \sum_{n=0}^\infty \alpha_n \ip{f}{\phi_n}\psi_n, \; f \in D(H_{\phi, \psi}^\balpha)
\end{array}\right. .
$$
Then we have the following
\begin{align}\label{2.1} &\D_\phi:= \mbox{span}\{\phi_n\} \subset D(H_{\phi, \psi}^\balpha)\,\, ;\\
& \D_\psi:= \mbox{span}\{\psi_n\} \subset D(H_{\psi, \phi}^\balpha) \nonumber\,\, ;\\
& H_{\phi, \psi}^\balpha \phi_k =\alpha_k \phi_k, \; k=0,1, \ldots\,\, ;\label{2.2}\\
& H_{\psi, \phi}^\balpha \psi_k =\alpha_k \psi_k, \; k=0,1, \ldots\, \nonumber \,.
\end{align}
Hence, $H_{\phi, \psi}^\balpha$ and $H_{\psi, \phi}^\balpha$ are densely defined.
\begin{prop}\label{prop_2.1}The following statements hold.
\begin{itemize}
\item[(1)]$D(H_{\phi, \psi}^\balpha)=\left\{f \in \Hil; \sum_{n=0}^\infty |\alpha_n|^2 |\ip{f}{\psi_n}|^2 <\infty \right\}$, \\
$D(H_{\psi, \phi}^\balpha)=\left\{f \in \Hil; \sum_{n=0}^\infty |\alpha_n|^2 |\ip{f}{\phi_n}|^2 <\infty \right\}$.
\item[(2)] $H_{\phi, \psi}^\balpha$ and $H_{\psi, \phi}^\balpha$ are closed.
\item[(3)] $\left(H_{\phi, \psi}^\balpha\right)^*=H_{\psi, \phi}^\balphabar$, where $\balphabar=\{\overline{\alpha}_n\}$.
\item[(4)] $H_{\phi, \psi}^\balpha$ is bounded if and only if $H_{\psi, \phi}^\balpha$ is bounded and if and only if $\balpha$ is a bounded sequence.
In particular $H_{\phi, \psi}^{\bm 1}=H_{\psi, \phi}^{\bm 1}=I$, where ${\bm 1}$ is the sequence constantly equal to $1$.
\end{itemize}
\end{prop}
\begin{proof} We will prove the statements (1)-(4) for $H_{\phi, \psi}^\balpha$.

(1): \; Since $T$ is invertible and has bounded inverse, there exist positive constants $\gamma_1, \gamma_2$ such that
$$ \gamma_1 \|f\| \leq \|Tf \| \leq \gamma_2 \|f\|, \forall f \in \Hil.$$
Hence, taking into account the equality
$$T\left( \sum_{k=n}^m \alpha_k \ip{f}{\psi_k}e_k\right)= \sum_{k=n}^m \alpha_k \ip{f}{\psi_k}\phi_k$$
we have

\begin{align}\label{2.3}\gamma_1^2 \sum_{k=n}^m |\alpha_k|^2 |\ip{f}{\psi_k}|^2&= \gamma_1^2\left\| \sum_{k=n}^m \alpha_k \ip{f}{\psi_k}e_k\right\|^2 \leq \left\|  \sum_{k=n}^m \alpha_k \ip{f}{\psi_k}\phi_k\right\|^2\\ & \leq
\gamma_2^2\left\| \sum_{k=n}^m \alpha_k \ip{f}{\psi_k}e_k\right\|^2 =\gamma_2^2 \sum_{k=n}^m |\alpha_k|^2 |\ip{f}{\psi_k}|^2,\nonumber
\end{align}
which shows that $f\in D(H_{\phi, \psi}^\balpha)$ if and only if $ \sum_{n=0}^\infty |\alpha_n|^2 |\ip{f}{\psi_n}|^2 <\infty$.

(2): \; Let $\{f_n\}$ be an arbitrary sequence in $D(H_{\phi, \psi}^\balpha)$ such that $f_n \to f$ and $H_{\phi, \psi}^\balpha f_n \to g$.
Then, for every $\epsilon >0$ there exists $N\in {\mb N}$ such that
$$ \left\|\sum_{k=0}^\infty \alpha_k \ip{f_n-f_m}{\psi_k}\phi_k\right\|<\epsilon, \quad \forall n,m\geq N.$$
By \eqref{2.3}, for all $M \in {\mb N}$ and $n,m \geq N$,
$$ \gamma_1^2 \sum_{k=0}^M |\alpha_k|^2 |\ip{f_n-f_m}{\psi_k}|^2 <\epsilon^2.$$
Hence, letting $m \to \infty$,
\begin{equation}\label{2.4} \gamma_1^2 \sum_{k=0}^M |\alpha_k|^2 |\ip{f_n-f}{\psi_k}|^2 \leq\epsilon^2, \quad \forall n\in {\mb N}. \end{equation}
Therefore,
\begin{align*} \gamma_1^2 \sum_{k=0}^M |\alpha_k|^2 |\ip{f}{\psi_k}|^2 &\leq 2 \gamma_1^2 \sum_{k=0}^M |\alpha_k|^2 |\ip{f-f_N}{\psi_k}|^2 + 2 \gamma_1^2 \sum_{k=0}^M |\alpha_k|^2 |\ip{f_N}{\psi_k}|^2 \\ &\leq 2\gamma_1^2 \epsilon^2 + 2 \gamma_1^2 \sum_{k=0}^M |\alpha_k|^2 |\ip{f_N}{\psi_k}|^2
 \end{align*}
and, letting $M \to \infty$,
$$ \gamma_1^2 \sum_{k=0}^\infty |\alpha_k|^2 |\ip{f}{\psi_k}|^2 \leq  2\gamma_1^2 \epsilon^2 + 2 \gamma_1^2 \sum_{k=0}^\infty |\alpha_k|^2 |\ip{f_N}{\psi_k}|^2<\infty.$$
This implies that $f \in D(H_{\phi, \psi}^\balpha)$.

Moreover, by \eqref{2.4}, we get
$$
\|H_{\phi, \psi}^\balpha f_n - H_{\phi, \psi}^\balpha f\|  \leq \gamma_2^2 \sum_{k=0}^\infty |\alpha_k|^2 |\ip{f_n-f}{\psi_k}|^2 \leq \frac{\gamma_2^2}{\gamma_1^2}\epsilon^2, \quad \forall n\geq N.
$$
Hence $$ \lim_{n\to \infty}H_{\phi, \psi}^\balpha f_n =H_{\phi, \psi}^\balpha f.$$
Thus, $H_{\phi, \psi}^\balpha$ is closed.

(3): \; It is easy to show that $\sum_{n=0}^\infty \overline{\alpha}_n \psi_n \otimes \overline{\phi_n} \subseteq (H_{\phi, \psi}^\balpha)^*$. Conversely, let $ g \in D((H_{\phi, \psi}^\balpha)^*)$; then there exists $h \in \Hil$ such that

$$ \ip{\sum_{n=0}^\infty {\alpha}_n (\phi_n \otimes \overline{\psi_n})f}{g}=\ip{f}{h}, \quad \forall f \in D((H_{\phi, \psi}^\balpha).$$

  By \eqref{2.1} and \eqref{2.2}, $\D_\phi \subseteq  D(H_{\phi, \psi}^\balpha)$ and $H_{\phi, \psi}^\balpha \phi_k =\alpha_k \phi_k$, $k=0,1, \ldots \,$ .
  Thus, $\ip{\alpha_k \phi_k}{g}= \ip{\phi_k}{h}$, $k=0,1, \ldots \,$.
  Hence
  $$ \sum_{k=0}^\infty |\alpha_k|^2 |\ip{\phi_k}{g}|^2=\sum_{k=0}^\infty  |\ip{\phi_k}{h}|^2=\sum_{k=0}^\infty  |\ip{e_k}{T^*h}|^2=\|T^*h\|^2.$$
  This implies that $g \in D(H_{\psi,\phi}^\balphabar)$.

(4): \; This is almost trivial.

In very similar way one can prove (1)-(4) for $H_{\psi, \phi}^\balpha$. This completes the proof.
\end{proof}

\medskip

If $\bbeta:=\{\beta_n\}$ is a sequence of complex numbers, we can define two other operators
$$S_\phi^{\bbeta}= \sum_{n=0}^\infty \beta_n \phi_n \otimes \overline{\phi}_n$$
and
$$S_\psi^\bbeta=\sum_{n=0}^\infty \beta_n \psi_n \otimes \overline{\psi}_n$$
as follows:

$$ \left\{\begin{array}{l} D(S_\phi^{\bbeta})=\left\{f \in \Hil; \sum_{n=0}^\infty \beta_n \ip{f}{\phi_n}\phi_n \mbox{ exists in }\Hil \right\}\\
{ } \\
S_\phi^{\bbeta} f= \sum_{n=0}^\infty \beta_n \ip{f}{\phi_n}\phi_n, \; f \in D(S_\phi^{\bbeta})
\end{array}\right.
$$

$$ \left\{\begin{array}{l} D(S_\psi^{\bbeta})=\left\{f \in \Hil; \sum_{n=0}^\infty \beta_n \ip{f}{\psi_n}\psi_n \mbox{ exists in }\Hil \right\}\\
{ } \\
S_\psi^{\bbeta} f= \sum_{n=0}^\infty \beta_n \ip{f}{\psi_n}\psi_n, \; f \in D(H_{S_\psi^{\bbeta}})
\end{array}\right. .
$$
It is clear that
\begin{align}\label{2.5} &\D_\psi\subset D(S_\phi^{\bbeta})\mbox{ and }S_\phi^{\bbeta} \psi_k =\beta_k \phi_k, \; k=0,1, \ldots\,\, ;\\
& \D_\phi\subset D(S_\psi^{\bbeta}) \mbox{ and }  S_\phi^{\bbeta} \phi_k =\beta_k \psi_k, \; k=0,1, \ldots\,
\label{2.6}
\end{align}
Hence, $S_\phi^{\bbeta}$ and $S_\psi^{\bbeta}$ are densely defined, and the following results can be established:
\begin{prop} \label{prop_2.2}
The following statements hold.
\begin{itemize}
\item[(1)]$D(S_\phi^{\bbeta})=\left\{f \in \Hil; \sum_{n=0}^\infty |\beta_n|^2 |\ip{f}{\phi_n}|^2 <\infty \right\}= D(H_{\psi, \phi}^\bbeta)$, \\
$D(S_\psi^{\bbeta})=\left\{f \in \Hil; \sum_{n=0}^\infty |\beta_n|^2 |\ip{f}{\psi_n}|^2 <\infty \right\}=D(H_{\phi, \psi}^\bbeta)$.
\item[(2)] $S_\phi^{\bbeta}$ and $S_\psi^{\bbeta}$ are closed.
\item[(3)] $\left(S_\phi^{\bbeta}\right)^*=S_\phi^{\bbetabar}$ and $\left(S_\psi^{\bbeta}\right)^*=S_\psi^{\bbetabar}$, where $\bbetabar=\{\overline{\beta_n}\}$.
\item[(4)] If $\{\beta_n\} \subset {\mb R}$ (respectively,  $\{\beta_n\} \subset {\mb R}^+$) then $S_\phi^{\bbeta}$ and $S_\psi^{\bbeta}$ are self-adjoint (respectively, positive self-adjoint). Furthermore, $S_\phi^{\bbeta}$ is bounded if and only if $S_\psi^{\bbeta}$ is bounded and if and only if $\bbeta$ is a bounded sequence.
\item[(4)] If, $\bbeta={\bm 1}$, where, as before, ${\bm 1}$ denotes the sequence constantly equal to $1$, then $S_\phi:=S_\phi^{\bm 1}$ and $S_\psi:=S_\psi^{\bm 1}$ are bounded positive self-adjoint operators on $\Hil$ and they are inverses of each other, that is $S_\phi= (S_\psi)^{-1}$, and $S_\phi =TT^*$, $S_\psi = (T^{-1})^*T^{-1}$.
\end{itemize}
\end{prop}

The proof of the Proposition is similar to that of the previous one, and will not be repeated.

Of course, we can construct the pair of operators $H_{\phi, \psi}^\balpha$, $H_{\psi, \phi}^\balpha$ and the pair of operators $S_\phi^{\bbeta}$, $S_\psi^{\bbeta}$ corresponding to different sequences $\balpha=\{\alpha_n\}$ and $\bbeta=\{\beta_n\}$ and study the interplay between them. In particular, we will take $\bbeta={\bm 1}$. This is particularly interesting since the relations between $H_{\phi, \psi}^\balpha$ , $H_{\psi, \phi}^\balpha$, $S_\phi$ and $S_\psi$ are given as follows:
\begin{prop}\label{prop_2.3} The following equalities hold:
\be\label{2add1}\left\{\begin{array}{l}
 S_\psi H_{\phi, \psi}^\balpha= H_{\psi, \phi}^\balpha S_\psi =S_\psi^\balpha,\\
 S_\phi H_{\psi, \phi}^\balpha= H_{\phi, \psi}^\balpha S_\phi=S_\phi^\balpha.
\end{array}\right. \en
\end{prop}
\begin{proof} By Proposition \ref{prop_2.2} we have
$D(H_{\phi, \psi}^\balpha)=D(S_\psi^\balpha)$ and $D(H_{\psi, \phi}^\balpha)=D(S_\phi^\balpha)$. Moreover, from Proposition \ref{prop_2.1},
$$ S_\psi f \in D(H_{\psi, \phi}^\balpha) \Leftrightarrow \sum_{n=0}^\infty |\alpha_n|^2 |\ip{S_\psi f}{\phi_n}|^2<\infty \Leftrightarrow
\sum_{n=0}^\infty |\alpha_n|^2 |\ip{f}{\psi_n}|^2<\infty \Leftrightarrow f \in D(S_\psi^\balpha).$$
Thus,
$$ D(S_\psi H_{\phi, \psi}^\balpha)=D(H_{\psi, \phi}^\balpha S_\psi) = D(S_\psi^\balpha).$$
It is easily seen that $S_\psi H_{\phi, \psi}^\balpha f= H_{\psi, \phi}^\balpha S_\psi f = S_\psi^\balpha f$, for every $f \in D(S_\psi^\balpha)$.
Hence,
$$S_\psi H_{\phi, \psi}^\balpha= H_{\psi, \phi}^\balpha S_\psi =S_\psi^\balpha.$$
In similar way one proves that
$$ S_\phi H_{\psi, \phi}^\balpha= H_{\phi, \psi}^\balpha S_\phi=S_\phi^\balpha,$$ as we had to prove.
\end{proof}

As shown in Proposition \ref{prop_2.1}, even if $\balpha=\{\alpha_n\}\subset {\mb R}$ the operator $H_{\phi, \psi}^\balpha$ is not necessarily self-adjoint.
The equality $H_{\psi, \phi}^\balpha S_\psi f =S_\psi H_{\phi, \psi}^\balpha f$, for every $f \in D(H_{\phi, \psi}^\balpha)$, stated in Proposition \ref{prop_2.3}, implies that the operator $H_{\phi, \psi}^\balpha$ is {\em quasi-hermitian} in the sense of \cite{jpctsmall}. Roughly speaking, a quasi-hermitian operator is an operator that can be made hermitian by changing the inner product of the space by means of a bounded {\em metric} operator $G$; i.e.,  $G$ is a bounded, strictly positive operator with bounded inverse. These are exactly the properties that $S_\psi$ and $S_\phi$ enjoy, under the assumptions we are adopting in this paper.  Metric operators define in the Hilbert space $\Hil$ where they act a new inner product, which gives rise to a topology equivalent to the original one of $\Hil$ but with a different {\em distance}. This situation changes deeply, as it is well explained in \cite{siegl}, if one extends the definition of metric operator by including the possibility that the inverse is not necessarily bounded or even that both the operator and its inverse are not bounded (in both cases the original Hilbert space is moved to another one). We refer to \cite{jpctsmall} for a detailed discussion.

Coming back to the situation under consideration, from \cite[Prop. 3.12]{jpctsmall} we get the results stated below. Let us denote by $\Hil(S)$ the Hilbert space obtained by defining on $\Hil$ the inner product $\ip{\cdot}{\cdot}_S$ given by $\ip{f}{g}_S= \ip{S_\psi f}{ g}$, $f,g \in \Hil$.
The quasi-hermitianness of $H_{\phi, \psi}^\balpha$ then implies that $H_{\phi, \psi}^\balpha$ is symmetric with respect to this new inner product; i.e.,
$
\ip{H_{\phi, \psi}^\balpha f}{g}_S=\ip{f}{H_{\phi, \psi}^\balpha g}_S,
$ for $f,g \in D(H_{\phi, \psi}^\balpha)$.
This is, in a certain sense, not surprising, since the set of eigenvectors of $H_{\phi, \psi}^\balpha$, $\F_\phi$, is an ONB in $\Hil$, when endowed with the  inner product $\ip{\cdot}{\cdot}_S$; indeed, $\left<\phi_n,\phi_m\right>_S=\delta_{n,m}$.
\vspace{2mm}

The hermitianness of $H_{\phi, \psi}^\balpha$ in $\Hil(S)$ is however of little use if one wants to make use of the powerful spectral theory for self-adjoint operators.
For this reason it is very convenient to have at hand conditions for the self-adjointness of $H_{\phi, \psi}^\balpha$ in $\Hil(S)$. By  \cite[Prop. 3.12]{jpctsmall}, the self-adjointness of $H_{\phi, \psi}^\balpha$ in $\Hil(S)$ is equivalent to the self-adjointness of the operator ${\sf h}_{\phi, \psi}^\balpha= S_\psi^{1/2} H_{\phi, \psi}^\balpha S_\phi^{1/2}$ in $\Hil$. As we shall see below, ${\sf h}_{\phi, \psi}^\balpha$ is actually self-adjoint when $\{\alpha_n\}\subset {\mb R}$. Hence $H_{\phi, \psi}^\balpha$ is self-adjoint when regarded as an operator in $\Hil(S)$.

Coming back to the general case, let $\balpha=\{\alpha_n\}$ be a sequence of complex numbers. We define the operators ${\sf h}_{\phi, \psi}^\balpha$ and ${\sf h}_{\psi, \phi}^\balpha$ as follows:
\be\label{2add2}\left\{\begin{array}{l}
 {\sf h}_{\phi, \psi}^\balpha =S_\psi^{1/2} H_{\phi, \psi}^\balpha S_\phi^{1/2},\\
{\sf h}_{\psi, \phi}^\balpha = S_\phi^{1/2} H_{\psi, \phi}^\balpha S_\psi^{1/2}.
\end{array}\right. \en

Then,

\begin{prop}\label{prop 2.4} The following statements hold:
\begin{itemize}
\item[(1)] $D({\sf h}_{\phi, \psi}^\balpha)=\{ S_\psi^{1/2}f; f \in D(H_{\phi, \psi}^\balpha)\}$,\\
$D({\sf h}_{\psi, \phi}^\balpha)=\{ S_\phi^{1/2}f; f \in D(H_{\psi, \phi}^\balpha)\}$ \\
and they are dense in $\Hil$.
\item[(2)] $({\sf h}_{\phi, \psi}^\balpha)^* = {\sf h}_{\psi, \phi}^\balphabar$.
\item[(3)] If $\{\alpha_n\}\subset {\mb R}$, then ${\sf h}_{\phi, \psi}^\balpha$ is self-adjoint.
\end{itemize}

\end{prop}
\begin{proof}

(1): \; This is a simple consequence of the equality $S_\phi^{-1/2}=S_\psi^{1/2}$.

(2): \; By Proposition \ref{prop_2.1}, we have
$$({\sf h}_{\phi, \psi}^\balpha)^*\supset S_\phi^{1/2}(H_{\phi, \psi}^\balpha)^*S_\psi^{1/2}= S_\phi^{1/2}(H_{\psi, \phi}^\balpha)^*S_\psi^{1/2}={\sf h}_{\psi, \phi}^\balphabar.$$
Conversely, let $g \in D(({\sf h}_{\phi, \psi}^\balpha)^*)$. Then, by (1) we have
$$ \ip{{\sf h}_{\phi, \psi}^\balpha S_\psi^{1/2} f}{g} = \ip{S_\psi^{1/2} f}{({\sf h}_{\phi, \psi}^\balpha)^*g}, \quad \forall f \in D(H_{\phi, \psi}^\balpha).$$
This implies that
$$\ip{H_{\phi, \psi}^\balpha f}{S_\psi^{1/2} g}= \ip{f}{S_\psi^{1/2}({\sf h}_{\phi, \psi}^\balpha)^*g}.$$
Hence, we have
\begin{align*}& S_\psi^{1/2} g \in D((H_{\phi, \psi}^\balpha)^*) \mbox { and }\\ &(H_{\phi, \psi}^\balpha)^*S_\psi^{1/2} g = S_\psi^{1/2} ({\sf h}_{\phi, \psi}^\balpha)^*g.
\end{align*}
This in turn implies that
$$ ({\sf h}_{\phi, \psi}^\balpha)^* \subset S_\phi^{1/2}(H_{\phi, \psi}^\balpha)^*S_\psi^{1/2}=S_\phi^{1/2}(H_{\psi, \phi}^\balphabar)S_\psi^{1/2}={\sf h}_{\psi, \phi}^\balphabar.$$

(3): \; Let $\{\alpha_n\}\subset {\mb R}$. For every $f \in D(H_{\phi, \psi}^\balpha)$, we have, making use of Proposition \ref{prop_2.3},
$$ S_\psi^{1/2}f= S_\phi^{1/2}S_\psi f \in S_\phi^{1/2}D(H_{\psi, \phi}^\balpha)=D({\sf h}_{\psi, \phi}^\balpha)=D(({\sf h}_{\phi, \psi}^\balpha)^*).$$
Hence
$$D({\sf h}_{\phi, \psi}^\balpha)\subset D(({\sf h}_{\phi, \psi}^\balpha)^*)$$
and
$$ {\sf h}_{\phi, \psi}^\balpha S_\psi^{1/2}f= S_\psi^{1/2}H_{\phi, \psi}^\balpha f= S_\phi^{1/2}S_\psi H_{\phi, \psi}^\balpha f = S_\phi^{1/2}H_{\psi, \phi}^\balpha S_\psi^{1/2}S_\psi^{1/2}f = ({\sf h}_{\phi, \psi}^\balpha)^*  S_\psi^{1/2}f,$$
for every $f \in D(H_{\phi, \psi}^\balpha)$. Hence, ${\sf h}_{\phi, \psi}^\balpha\subset ({\sf h}_{\phi, \psi}^\balpha)^*$.

Conversely, take an arbitrary $g \in D(H_{\psi, \phi}^\balpha)$. Then,
$$ S_\phi^{1/2} g=  S_\psi^{1/2}S_\phi g \in S_\psi^{1/2} D(H_{\phi, \psi}^\balpha)=D({\sf h}_{\phi, \psi}^\balpha), $$
and so $D(({\sf h}_{\phi, \psi}^\balpha)^*)= D({\sf h}_{\psi, \phi}^\balpha)\subset D({\sf h}_{\phi, \psi}^\balpha)$. Hence, ${\sf h}_{\phi, \psi}^\balpha= ({\sf h}_{\phi, \psi}^\balpha)^*$. This completes the proof.

\end{proof}

Equations (\ref{2add2}) produce $S_\phi^{1/2} {\sf h}_{\phi, \psi}^\balpha= H_{\phi, \psi}^\balpha S_\phi^{1/2}$ and $S_\psi^{1/2} {\sf h}_{\psi, \phi}^\balpha= H_{\psi, \phi}^\balpha S_\psi^{1/2}$, which, together with the equalities in (\ref{2add1}), are intertwining relations between different operators having the same eigenvalues and related eigenvectors. This fact is well known in the physical literature, and we refer to \cite{intop} and references therein for some appearances of these kind of equations in concrete models.

\section{Generalized lowering and raising operators}
Following what we did in \cite{bit2012}, we now introduce generalized lowering and raising operators as follows:

\begin{align*} &A_{\phi, \psi}^{\bgamma}= \sum_{n=1}^\infty \gamma_n \phi_{n-1} \otimes \overline{\psi}_n\\
&A_{\psi, \phi}^\bgamma=\sum_{n=1}^\infty \gamma_n \psi_{n-1} \otimes \overline{\phi}_n\\
&B_{\phi, \psi}^{\bgamma}= \sum_{n=0}^\infty \gamma_{n+1} \phi_{n+1} \otimes \overline{\psi}_n\\
&B_{\psi, \phi}^\bgamma=\sum_{n=0}^\infty \gamma_{n+1} \psi_{n+1} \otimes \overline{\phi}_n.
\end{align*}

\begin{prop}\label{prop_2.5} 
The following statements hold.
\begin{itemize}
\item[(1)]$D(A_{\phi, \psi}^\bgamma)=\left\{f \in \Hil; \sum_{n=1}^\infty |\gamma_n|^2 |\ip{f}{\psi_n}|^2 <\infty \right\}=D(H_{\phi, \psi}^{\{\alpha_{n}\}}))$ \\
$D(A_{\psi, \phi}^\bgamma)=\left\{f \in \Hil; \sum_{n=1}^\infty |\gamma_n|^2 |\ip{f}{\phi_n}|^2 <\infty \right\}=D(H_{\psi, \phi}^{\{\alpha_{n}\}}))$\\
$D(B_{\phi, \psi}^\bgamma)=\left\{f \in \Hil; \sum_{n=0}^\infty |\gamma_{n+1}|^2 |\ip{f}{\psi_n}|^2 <\infty \right\}=D(H_{\phi, \psi}^{\{\alpha_{n+1}\}})$\\
$D(B_{\psi, \phi}^\bgamma)=\left\{f \in \Hil; \sum_{n=0}^\infty |\gamma_{n+1}|^2 |\ip{f}{\phi_n}|^2 <\infty \right\}=D(H_{\psi, \phi}^{\{\alpha_{n+1}\}})$
\item[(2)] $A_{\phi, \psi}^\bgamma$, $A_{\psi, \phi}^\bgamma$, $B_{\phi, \psi}^\bgamma$ and $B_{\psi, \phi}^\bgamma$ are densely defined closed operators in $\Hil$.
\item[(3)] $\left(A_{\phi, \psi}^\bgamma\right)^*=B_{\psi, \phi}^\bgammabar$ and
$\left(A_{\psi, \phi}^\bgamma\right)^*=B_{\phi, \psi}^\bgammabar$.
\end{itemize}
\end{prop}
The proof is similar to that of Proposition \ref{prop_2.1}.

Since the relevance of these operators relies essentially on their products, we give the following

\begin{prop}\label{addprop32}The following statements hold.
\begin{itemize}
\item[(1)] $D(A_{\phi, \psi}^\bgamma B_{\phi, \psi}^\bgamma) =\left\{ f\in \Hil; \sum_{n=0}^\infty |\gamma_{n+1}|^2 |\ip{f}{\psi_n}|^2 <\infty \mbox{ and }  \sum_{n=0}^\infty |\gamma_{n+1}|^4 |\ip{f}{\psi_n}|^2 <\infty\right\}$
    and

    $ A_{\phi, \psi}^\bgamma B_{\phi, \psi}^\bgamma \subset \sum_{n=0}^\infty \gamma_{n+1}^2 \phi_{n} \otimes \overline{\psi}_n =H_{\phi, \psi}^{\{\gamma_{n+1}^2\}}$.

    \item[(2)] $D(B_{\phi, \psi}^\bgamma A_{\phi, \psi}^\bgamma) =\left\{ f\in \Hil; \sum_{n=0}^\infty |\gamma_{n}|^2 |\ip{f}{\psi_n}|^2 <\infty \mbox{ and }  \sum_{n=0}^\infty |\gamma_{n}|^4 |\ip{f}{\psi_n}|^2 <\infty\right\}$
    and

    $ B_{\phi, \psi}^\bgamma A_{\phi, \psi}^\bgamma \subset \sum_{n=0}^\infty \gamma_{n }^2 \phi_{n} \otimes \overline{\psi}_n =H_{\phi, \psi}^{\{\gamma_{n}^2\}}$.
\end{itemize}
Suppose that $|\gamma_0| \leq |\gamma_1| \leq \cdots$\,. Then the following statements (3) and (4) hold:
\begin{itemize}
\item[(3)] $D(B_{\phi, \psi}^\bgamma A_{\phi, \psi}^\bgamma) =\left\{ f\in \Hil;  \sum_{n=0}^\infty |\gamma_{n}|^4 |\ip{f}{\psi_n}|^2 <\infty\right\}= D(H_{\phi, \psi}^{\{\gamma_{n}^2\}})$;\\
    $D(A_{\phi, \psi}^\bgamma B_{\phi, \psi}^\bgamma) =\left\{ f\in \Hil; \ \sum_{n=0}^\infty |\gamma_{n+1}|^4 |\ip{f}{\psi_n}|^2 <\infty\right\}=D(H_{\phi, \psi}^{\{\gamma_{n+1}^2\}})$;\\
    $B_{\phi, \psi}^\bgamma A_{\phi, \psi}^\bgamma \subset A_{\phi, \psi}^\bgamma B_{\phi, \psi}^\bgamma$;
    \\$B_{\phi, \psi}^\bgamma A_{\phi, \psi}^\bgamma = \sum_{n=0}^\infty \gamma_{n}^2 \phi_{n} \otimes \overline{\psi}_n =H_{\phi, \psi}^{\{\gamma_{n}^2\}}$;\\
    $A_{\phi, \psi}^\bgamma B_{\phi, \psi}^\bgamma = \sum_{n=0}^\infty \gamma_{n+1}^2 \phi_{n} \otimes \overline{\psi}_n =H_{\phi, \psi}^{\{\gamma_{n+1}^2\}}$.
    \item[(4)] $A_{\phi, \psi}^\bgamma B_{\phi, \psi}^\bgamma f - B_{\phi, \psi}^\bgamma A_{\phi, \psi}^\bgamma f =H_{\phi, \psi}^{\{\gamma_{n+1}^2\}}f - H_{\phi, \psi}^{\{\gamma_{n}^2\}}f = \left( \sum_{n=0}^\infty (\gamma_{n+1}^2 - \gamma_{n}^2)\phi_{n} \otimes \overline{\psi}_n
    \right)f$, for every $f \in D(H_{\phi, \psi}^{\{\gamma_{n+1}^2\}})$.
\end{itemize}
\end{prop}

\begin{proof} We put for shortness $A:=A_{\phi, \psi}^\bgamma$ and $ B:= B_{\phi, \psi}^\bgamma$.

(1): \; By (1) and (2) of Proposition \ref{prop_2.5} we have
\begin{align*}D(AB)&= \left\{ f\in \Hil; \sum_{n=0}^\infty |\gamma_{n+1}|^2 |\ip{f}{\psi_n}|^2 <\infty \mbox{ and }  \sum_{n=1}^\infty |\gamma_{n}|^2 |\ip{Bf}{\psi_n}|^2 <\infty\right\}\\
&=\left\{ f\in \Hil; \sum_{n=0}^\infty |\gamma_{n+1}|^2 |\ip{f}{\psi_n}|^2 <\infty \mbox{ and }  \sum_{n=0}^\infty |\gamma_{n+1}|^4 |\ip{f}{\psi_n}|^2 <\infty\right\}\\
& \subset \left\{ f\in \Hil; \sum_{n=0}^\infty |\gamma_{n+1}|^2 |\ip{f}{\psi_n}|^2 <\infty\right\}\\
&=D(\sum_{n=0}^\infty \gamma_{n+1}^2 \phi_{n} \otimes \overline{\psi}_n)
\end{align*}
and
$$ ABf = \left(\sum_{n=0}^\infty \gamma_{n+1}^2 \phi_{n} \otimes \overline{\psi}_n\right)f, \quad \forall f\in D(AB).$$
Hence, $$AB \subset \sum_{n=0}^\infty \gamma_{n+1}^2 \phi_{n} \otimes \overline{\psi}_n.$$

The statements for $BA$ in (2) are proved in similar way.

(3): \; If $\{\gamma_{n}\}$ is bounded then both $A$ and $B$ are bounded, and so (3) and (4) hold. As for the general case, we begin with putting
$ N_1 = \max \{n \in {\mb N}\cup \{0\}; |\alpha _n|\leq 1\}$.

Let us suppose that $\sum_{n=0}^\infty |\gamma_{n}|^4 |\ip{f}{\psi_n}|^2 <\infty$. Then we have
$$\sum_{n=0}^\infty |\gamma_{n}|^2|\ip{f}{\psi_n}|^2 \leq \sum_{n=0}^{N_1} |\gamma_{n}|^2|\ip{f}{\psi_n}|^2 +\sum_{n=N_1+1}^\infty |\gamma_{n}|^4 |\ip{f}{\psi_n}|^2 <\infty.$$
This implies that
$$ D(BA)=\left\{ f \in \Hil;\sum_{n=0}^\infty |\gamma_{n}|^4 |\ip{f}{\psi_n}|^2 <\infty \right\}$$ and $$BA=\sum_{n=0}^\infty \gamma_{n}^2 \phi_{n} \otimes \overline{\psi}_n.$$
In similar way, we have
$$ D(AB)=\left\{ f\in \Hil; \ \sum_{n=0}^\infty |\gamma_{n+1}|^4 |\ip{f}{\psi_n}|^2 <\infty\right\}$$
and
$$ AB= \sum_{n=0}^\infty \gamma_{n+1}^2 \phi_{n} \otimes \overline{\psi}_n.$$
Clearly, $D(AB)\subset D(BA)$.

(4): \; This follows easily from (3).
\end{proof}

It is clear that the lowering and raising operators considered above constitute a generalization of annihilation and creation operators of Quantum Mechanics and, as in that case, they can be used to factorize the original hamiltonian. In particular, if the sequence $\balpha=\{\alpha_n\}$ introduced in the previous section is such that $0=\alpha_0<\alpha_1<\alpha_2<\ldots$, and if we take $\gamma_n=\sqrt{\alpha_n}$ here, the sequence $\bgamma$ satisfies the assumptions of Proposition \ref{addprop32}, and we find, for instance that $B_{\phi, \psi}^\bgamma A_{\phi, \psi}^\bgamma = \sum_{n=0}^\infty \alpha_{n} \phi_{n} \otimes \overline{\psi}_n =H_{\phi, \psi}^{\balpha}$: then $H_{\phi, \psi}^{\balpha}$ can be factorized. Moreover, the commutation relation deduced in (4) is a stronger version of the situation considered  \cite{bit2012} and the results obtained there apply.

\section{Examples and applications}

\subsection{A {\em no-go} example}
Recently, in \cite{siegl}, the authors proved that, for the well known cubic hamiltonian $H=p^2+ix^3$, whose eigenvalues are real and positive, see \cite{tateo}, it is possible to find a bounded metric operator { which {\em transforms} $H$ into a self-adjoint operator ${\sf h}$, but this metric operator cannot have a bounded inverse}. On the other hand, as we have seen, in our settings both $S_\psi$ and $S_\phi$ are both bounded. The obvious conclusion is, therefore, that the basis of eigenvectors of $H$ is not a Riesz basis.

\subsection{A finite dimensional, one-parameter example}
Finite dimensional examples are quite useful to clarify several aspects of the general framework one is considering. In particular, in the context of PT quantum mechanics, this kind of examples are very common, see, for instance, \cite{findimex}. This is also useful in order to avoid problems with unbounded operators, which are clearly absent in this case.

Let $$H=\left(
\begin{array}{ccc}
 -\frac{3 \cos[t]}{5} & \frac{\sin[t]}{2} & \frac{3 \cos[t]}{5} \\
 -\frac{2 \sin[t]}{5} & \cos[t] & \frac{2 \sin[t]}{5} \\
 -\frac{18 \cos[t]}{5} & 3 \sin[t] & \frac{18 \cos[t]}{5}
\end{array}
\right)
$$
be a one-parameter, non self-adjoint operator in $\Hil={\Bbb C}^3$, our hamiltonian. Notice that $H\neq H^\dagger$ for any value of $t$, which we take, for the time being, in $[0,2\pi[$. The eigenvalues of $H$ are $\epsilon_0=0$, $\epsilon_1=2\cos(t)-1$, $\epsilon_2=2\cos(t)+1$. Then, if $t\in I:=[0,\frac{\pi}{3}[\cup[\frac{5\pi}{3},2\pi]$, the eigenvalues are simple and growing: $\epsilon_0<\epsilon_1<\epsilon_2$. The related eigenvectors are
$$
\phi_0=\left(
\begin{array}{c}
 1 \\
 0 \\
 1
\end{array}
\right),\qquad \phi_1= \left(
\begin{array}{c}
   {\sin[t/2]}\\
   {-2 \cos[t/2]}\\
   {6 \sin[t/2]}
  \end{array}
\right),\qquad \phi_2=\left(
\begin{array}{c}
   {\cos[t/2]}\\
   {2 \sin[t/2]}\\
   {6 \cos[t/2]}
  \end{array}
\right).
$$
Taking now as $\E$ the canonical basis in $\Hil$, the matrix $T$ introduced in Section 1 can be easily identified:
$$
T=\left(
\begin{array}{ccc}
 1 & \sin\left[\frac{t}{2}\right] &\cos\left[\frac{t}{2}\right] \\
 0 & -2 \cos\left[\frac{t}{2}\right] & 2 \sin\left[\frac{t}{2}\right] \\
 1 & 6\sin\left[\frac{t}{2}\right] & 6\cos\left[\frac{t}{2}\right]
\end{array}
\right),
$$
and the vectors of $\F_\psi$, $\psi_n=(T^*)^{-1}e_n$, turns out to be
$$
\psi_0=\frac{1}{5}\left(
\begin{array}{c}
 6 \\
 0 \\
 -1
\end{array}
\right),\qquad \psi_1= \left(
\begin{array}{c}
   -\frac{1}{5}{\sin[t/2]}\\
   {-\frac{1}{2} \cos[t/2]}\\
   {\frac{1}{5} \sin[t/2]}
  \end{array}
\right),\qquad \psi_2=\left(
\begin{array}{c}
   {-\frac{1}{5}\cos[t/2]}\\
   {\frac{1}{2} \sin[t/2]}\\
   {\frac{1}{5} \cos[t/2]}
  \end{array}
\right).
$$
Direct computations show that $\left<\phi_n,\psi_m\right>=\delta_{n,m}$, and that
$$
S_\phi=\left(
\begin{array}{ccc}
 2 & 0 & 7 \\
 0 & 4 & 0 \\
 7 & 0 & 37
\end{array}
\right),
\qquad
S_\psi=
\left(
\begin{array}{ccc}
 \frac{37}{25} & 0 & -\frac{7}{25} \\
 0 & \frac{1}{4} & 0 \\
 -\frac{7}{25} & 0 & \frac{2}{25}
\end{array}
\right).
$$
These matrices, which are manifestly self-adjoint, are also positive and we have $S_\phi=S_\psi^{-1}$, $S_\phi=TT^*$ and $H^* S_\psi=S_\psi H$, as expected. This intertwining relation holds all over $\Hil$, clearly. Moreover, since $\|T\|=\sqrt{\frac{1}{2}(39+7\sqrt{29})}$ and $\|T^{-1}\|=\frac{1}{5}\|T\|$, we conclude, as it was already evident, that $\F_\phi$ and $\F_\psi$ are Riesz bases.

Interestingly enough, the positive square root of  $S_\phi$,
$$
S_\phi^{1/2}=\left(
\begin{array}{ccc}
 1 & 0 & 1 \\
 0 & 2 & 0 \\
 1 & 0 & 6
\end{array}
\right),
$$
does not coincide with $T$ for any possible value of the parameter $t\in I$. This, in a sense, is not surprising since $S_\phi^{1/2}$ is, by construction, self-adjoint, while $T$ is not. Still we find
$$
h=S_\psi^{1/2}HS_\phi^{1/2}=\left(
\begin{array}{ccc}
 0 & 0 & 0 \\
 0 & \cos[t] & \sin[t] \\
 0 & \sin[t] & 3\cos[t]
\end{array}
\right),
$$
which is self-adjoint. As for the raising and lowering operators, they are found to be
{\footnotesize$$
A=\left(
\begin{array}{ccc}
 -\frac{1}{5} \left(\sqrt{\epsilon_1}+\cos\left[\frac{t}{2}\right] \sqrt{\epsilon_2}\right) \sin\left[\frac{t}{2}\right] & \frac{1}{2} \left(-\cos\left[\frac{t}{2}\right] \sqrt{\epsilon_1}+\sqrt{\epsilon_2} \sin\left[\frac{t}{2}\right]^2\right) & \frac{1}{10} \left(2 \sqrt{\epsilon_1} \sin\left[\frac{t}{2}\right]+\sqrt{\epsilon_2} \sin[t]\right) \\
 \frac{1}{5} (1+\cos[t]) \sqrt{\epsilon_2} & -\frac{1}{2} \sqrt{\epsilon_2} \sin[t] & -\frac{1}{5} (1+\cos[t]) \sqrt{\epsilon_2} \\
 \frac{1}{5} \left(-\sqrt{\epsilon_1} \sin\left[\frac{t}{2}\right]-3 \sqrt{\epsilon_2} \sin[t]\right) & -\frac{1}{2} \cos\left[\frac{t}{2}\right] \sqrt{\epsilon_1}+3 \sqrt{\epsilon_2} \sin\left[\frac{t}{2}\right]^2 & \frac{1}{5} \left(\sqrt{\epsilon_1} \sin\left[\frac{t}{2}\right]+3 \sqrt{\epsilon_2} \sin[t]\right)
\end{array}
\right)
$$}
and
{\footnotesize$$
B=\left(
\begin{array}{ccc}
 \frac{6}{5} \sqrt{\epsilon_1} \sin\left[\frac{t}{2}\right]-\frac{1}{10} \sqrt{\epsilon_2} \sin[t] & -\frac{1}{4} (1+\cos[t]) \sqrt{\epsilon_2} & \frac{1}{10} \left(-2 \sqrt{\epsilon_1} \sin\left[\frac{t}{2}\right]+\sqrt{\epsilon_2} \sin[t]\right) \\
 -\frac{2}{5} \left(6 \cos\left[\frac{t}{2}\right] \sqrt{\epsilon_1}+\sqrt{\epsilon_2} \sin\left[\frac{t}{2}\right]^2\right) & -\frac{1}{2} \sqrt{\epsilon_2} \sin[t] & \frac{2}{5} \left(\cos\left[\frac{t}{2}\right] \sqrt{\epsilon_1}+\sqrt{\epsilon_2} \sin\left[\frac{t}{2}\right]^2\right) \\
 \frac{3}{5} \left(12 \sqrt{\epsilon_1} \sin\left[\frac{t}{2}\right]-\sqrt{\epsilon_2} \sin[t]\right) & -\frac{3}{2} (1+\cos[t]) \sqrt{\epsilon_2} & \frac{3}{5} \left(-2 \sqrt{\epsilon_1} \sin\left[\frac{t}{2}\right]+\sqrt{\epsilon_2} \sin[t]\right)
\end{array}
\right).
$$}
Despite of the complicated expressions of these operators, it is possible to check that $A\phi_0=0$, $A\phi_1=\sqrt{\epsilon_1}\phi_0$, $A\phi_2=\sqrt{\epsilon_2}\phi_1$, while $B\phi_0=\sqrt{\epsilon_1}\phi_1$, $B\phi_1=\sqrt{\epsilon_2}\phi_2$ and $B\phi_2=0$. As for their adjoints, we have $B^*\psi_0=0$, $B^*\psi_1=\sqrt{\epsilon_1}\psi_0$, $B^*\psi_2=\sqrt{\epsilon_2}\psi_1$, while $A^*\psi_0=\sqrt{\epsilon_1}\psi_1$, $A^*\psi_1=\sqrt{\epsilon_2}\psi_2$ and $A^*\psi_2=0$. Moreover, $H=BA$. Incidentally, $\F_\phi$ is also a set of eigenstates of the operator $\hat H=AB$, but with different eigenvalues $\tilde\epsilon_n$: $\tilde\epsilon_0=\epsilon_1$, $\tilde\epsilon_1=\epsilon_2$, $\tilde\epsilon_2=\epsilon_0$. Also, $\F_\psi$ is a set of eigenstates of the operator $\hat H^*$, with these same eigenvalues.

Finally we have
$$
a=S_\psi^{1/2}AS_\phi^{1/2}=\left(
\begin{array}{ccc}
 0 & -\cos\left[\frac{t}{2}\right] \sqrt{\epsilon_1} & \sqrt{\epsilon_1} \sin\left[\frac{t}{2}\right] \\
 0 & -\frac{1}{2} \sqrt{\epsilon_2} \sin[t] & -\frac{1}{2} (1+\cos[t]) \sqrt{\epsilon_2} \\
 0 & \sqrt{\epsilon_2} \sin\left[\frac{t}{2}\right]^2 & \frac{1}{2} \sqrt{\epsilon_2} \sin[t]
\end{array}
\right),
$$
and
$$
\Phi_0=S_\psi^{1/2}\phi_0\left(
\begin{array}{c}
 1 \\
 0 \\
 0
\end{array}
\right),\quad \Phi_1=S_\psi^{1/2}\phi_1= \left(
\begin{array}{c}
   0\\
   - \cos[t/2\\
   6 \sin[t/2]
  \end{array}
\right),\quad \Phi_2=S_\psi^{1/2}\phi_2=\left(
\begin{array}{c}
   0\\
    \sin[t/2]\\
   6 \cos[t/2]
  \end{array}
\right),
$$
which is an ON basis of eigenvectors of $h=a^* a$, with eigenvalues $\epsilon_n$. These vectors are also eigenstates of $\hat h=aa^*$, with eigenvalues $\tilde\epsilon_n$.


\subsection{An infinite dimensional example}

{ Let $\E=\{e_n,\, n\geq0\}$ be an orthonormal basis of $\Hil$ and $P=P^*=P^2$ an orthogonal projection. Let us put $T= I+iP$. Clearly $T$ is bounded and has bounded inverse; namely, $T^{-1}=I-\frac{i+1}{2}P$.
Then, as discussed in Section 2, if we put $\phi_n= Te_n$ and $\psi_n=(T^{-1})^* e_n$, $n=1,2,\ldots,$ the sets $\F_\phi=\{\phi_n,\,n\in {\mb N}\}$ and $\F_\psi=\{\psi_n,\,n\in {\mb N}\}$ are two biorthogonal Riesz bases.
So that if $\balpha=\{\alpha_n\}$ is a sequence of real numbers, as shown before,
the operator $H^\balpha_{\phi,\psi}$ can be defined. The operator $S_\phi=TT^*$ is then given by $S_\phi=I+P$. It is also easy to compute the positive square roots $S_\phi^{1/2}=I+(\sqrt{2}-1)P$ and $S_\psi^{1/2}=S_\phi^{-1/2}=I-\frac{2-\sqrt{2}}{2}P$.

In order to give a concrete example, we choose a particular form for $P$. Let $u=c_1e_1 +\cdots +c_Ne_N$ be a normalized linear combination of the first $N$ elements of the basis $\E$. We put $P:= u\otimes \overline{u}$. If we suppose that for at least two indices $j,k$, $1\leq j,k \leq N$, the coefficients $c_j$ and $c_k$ are both nonzero, then $P$ does not commute with all projections $e_k\otimes \overline{e_k}$, $k \in {\mb N}$. A direct computation shows that

$$\phi_n=\left\{ \begin{array}{ll}e_n+i\ip{e_n}{u}u = e_n+i \overline{c}_nu &\mbox{ if } n\leq N \\
e_n &\mbox{ if } n> N\end{array} \right.$$
and, similarly,
$$\psi_n=\left\{ \begin{array}{ll}e_n+\frac{i-1}{2}\ip{e_n}{u}u = e_n+\frac{i-1}{2} \overline{c}_nu &\mbox{ if } n\leq N \\
e_n &\mbox{ if } n> N\end{array} \right..$$
Then,
$$ H^\balpha_{\phi,\psi}= \sum_{k=1}^N \alpha_k \phi_k\otimes \overline{\psi}_k + \sum_{k=N+1}^\infty\alpha_k e_k\otimes \overline{e}_k .$$
Then, by Proposition \ref{prop 2.4}, the operator  ${\sf h}_{\phi, \psi}^\balpha =S_\psi^{1/2} H_{\phi, \psi}^\balpha S_\phi^{1/2}$ is self-adjoint and similar to $H_{\phi, \psi}^\balpha$. Thus they have the same spectrum. It is easy to check that $S_\psi^{1/2} (e_k\otimes \overline{e}_k) S_\phi^{1/2}= e_k\otimes \overline{e}_k$ if $k>N$, since $S_\psi^{1/2}$ acts nonidentically only on a finite-dimensional subspace of $\Hil$. Hence ${\sf h}_{\phi, \psi}^\balpha$ has the form
$${\sf h}_{\phi, \psi}^\balpha = \sum_{k=1}^N \alpha_k \zeta_k \otimes \overline{\zeta}_k +\sum_{k=N+1}^\infty \alpha_ke_k\otimes \overline{e_k}, $$
where $\{\zeta_1, \ldots, \zeta_N\}$ is an ON basis of the closed vector subspace spanned by $\{e_1, \ldots, e_N\}$, in general, different from the latter basis. It is worth remarking that in this case $$D({\sf h}_{\phi, \psi}^\balpha)= D(H^\balpha_{\phi,\psi})= \left\{ f\in \Hil;\, \sum_{k=1}^\infty|\alpha_k|^2|\ip{f}{e_k}|^2<\infty\right\}.$$

}

\bigskip

\vspace{3mm}

\section{Conclusions}
The non-self-adjoint hamiltonians we have considered here are probably the simplest ones (they are in a sense so much regular as operators defined by an ONB $\{e_n \}$ are, the crucial difference relying on the fact the latter are all normal operators). But as we said before, as far as we know, a detailed analysis of this interesting case was missing in the literature.

There are many questions that can be posed in this framework. The most interesting (and difficult) is very likely the following: under what conditions can a closed operator $H$ be constructed from a Riesz basis in the way we did in Section 2? We imagine there is no general answer to this problem as well as there is no general answer for the corresponding question in the case of self-adjoint operators, where only certain classes of them are known to have a discrete simple spectrum.

Another interesting question to be studied is similar to that considered in this paper but with a crucial difference:  assume that $H$ has, as before, a purely discrete simple spectrum but that  the corresponding  eigenvectors $\{\phi_n\}$ {\em do not} form a Riesz basis but just a basis or, even less, a $\D$-quasi basis in the sense of \cite{bagnew}: is it possible to reconstruct $H$ from the basis in the same spirit of what we did here for Riesz bases. Work on this matter is in progress and we hope to discuss this case in detail in a future paper.

\section*{Acknowledgements}
The authors would like to acknowledge financial support from the
   Universit\`a di Palermo and from the Fukuoka University.


\begin{thebibliography}{99}
\bibitem{nota1} As discussed in
 some recent papers, se \cite{bagnew,bagfring} for instance, this is a rather strong assumption, surely interesting from a mathematical point of view but not so {frequent} in concrete physical models. However, it may happen, and we refer to \cite{petr} for a concrete model admitting Riesz bases of eigenvectors.
\bibitem{nota2} If $\phi, \psi\in \Hil$, we denote, as usual, by $\phi\otimes \overline{\psi}$ the linear operator defined by $(\phi\otimes \overline{\psi})(f)=\ip{f}{\psi}{\phi},$ $f\in \Hil$.
\end{thebibliography}
\end{document}